\documentclass[a4paper,UKenglish]{lipics-v2016}

\usepackage{algorithm,algpseudocode}
\usepackage{algorithmicx}

\newcommand{\ou}{|{\it output}|}
\newcommand{\Or}{{\rm O}}
\newcommand{\order}{{\rm o}}
\newcommand{\ep}{\epsilon}
\newcommand{\ans}{{\it ans}}
\newcommand{\ac}{{\sf access}}
\newcommand{\bo}{{\sf bound}}
\newcommand{\ra}{{\sf rank}}
\newcommand{\sel}{{\sf select}}
\newcommand{\tims}{{\it IMS2}}
\newcommand{\ims}{{\it IMS}}
\newcommand{\dms}{{\it DMS}}
\newcommand{\dis}{\displaystyle}

\newtheorem{property}{Property}
 
\usepackage{microtype}

\title{Enumerating Range Modes}
\titlerunning{Enumerating Range Modes} 

\author[1]{Kentaro Sumigawa}
\author[2]{Sankardeep Chakraborty}
\author[3]{Kunihiko Sadakane}
\author[4]{Srinivasa Rao Satti}
\affil[1]{The University of Tokyo, Japan\\
  \texttt{kentaro\_sumigawa@me2.mist.i.u-tokyo.ac.jp}}
\affil[2]{RIKEN Center for Advanced Intelligence Project, Japan\\
  \texttt{sankar.chakraborty@riken.jp}}
 \affil[3]{The University of Tokyo, Japan\\
  \texttt{sada@mist.i.u-tokyo.ac.jp}}
  \affil[4]{Seoul National University, South Korea\\
  \texttt{ssrao@cse.snu.ac.kr}}
\authorrunning{Sumigawa, Chakraborty, Sadakane and Satti} 

\Copyright{Sumigawa, Chakraborty, Sadakane and Satti}

\subjclass{Dummy classification -- please refer to \url{http://www.acm.org/about/class/ccs98-html}}

\EventEditors{John Q. Open and Joan R. Acces}
\EventNoEds{2}
\EventLongTitle{42nd Conference on Very Important Topics (CVIT 2016)}
\EventShortTitle{CVIT 2016}
\EventAcronym{CVIT}
\EventYear{2016}
\EventDate{December 24--27, 2016}
\EventLocation{Little Whinging, United Kingdom}
\EventLogo{}
\SeriesVolume{42}
\ArticleNo{23}

\begin{document}

\maketitle

\begin{abstract}
We consider the range mode problem where given a sequence and
a query range in it, we want to find items with maximum frequency
in the range. We give time- and space- efficient algorithms for this problem. Our algorithms are efficient for small maximum frequency cases.
We also consider a natural generalization of the problem:
the range mode enumeration problem, for which there has been no
known efficient algorithms.  Our algorithms have query time 
complexities which is linear to the output size plus small terms.
\end{abstract}


\section{Introduction}
\label{sec:intro}

We consider the range mode problem, defined as follows.
\begin{definition}[Mode]
Given a non-empty multiset $S$,
$x\in S$ is said to be a \emph{mode} of $S$,
if its multiplicity is no smaller than those of any other elements.
\end{definition}

\begin{definition}[Range mode problem]
For a sequence $A[0...n-1]$ and a range $[l,r]$ of $A$ ($0\leq l \leq r<n$),
output any one of the modes of the multiset $\{A[l],A[l+1],\ldots,A[r-1],A[r]\}$.
\end{definition}
The problem has many applications in data mining and data analysis~\cite{DemaineLM02,FangSGMU98}. Moreover, there is a strong interest in theory community as well for this problem as it is related to the famous Boolean matrix multiplication and set intersection problem~\cite{Chan2014}.

In this paper, we consider the indexing version of the range mode problem.
That is, given a sequence of length $n$, we first construct a data structure,
called an \emph{index}.  Then given a query range $[l,r]$, we solve the query using
the index as well as the input.  The algorithm is measured by the index size (in bits) and query time complexity.
There are many existing work~\cite{Petersen08,Chan2014,PG09,GreveJLT10,DM11}\footnote{The papers~\cite{Chan2014,DM11} have the same title, but some of the results
in~\cite{DM11} do not appear in~\cite{Chan2014}.} and some of them are summarized in Table~\ref{table-rmo}.

Our first contribution is space-efficient indexes for the range mode problem,
for the case the maximum multiplicity $m$ of an item in the set is small.
Table~\ref{table-rmo} summarizes our results.  The one in Corollary~\ref{cor-loglog} has better time and space complexities than that of~\cite{DM11} with $\ep=0$,
which is also specialized for small $m$ and has
space complexity $\Or(nm\log n)$ bits and query time complexity $\Or(\log\log n)$.

Our second contribution is efficient indexes for the range mode enumeration problem,
defined as follows.
\begin{definition}[Range Mode Enumeration Problem]
Given a sequence $A[0...n-1]$ and a query range $[l,r]\ (0\leq l \leq r<n)$, output \emph{all} items with the largest number of occurrences
in the multiset $\{A[l],A[l+1],\ldots,A[r-1],A[r]\}$.
\end{definition}

Though the problem seems to be a natural generalization of the range mode problem,
there has been no existing work.
A related and important problem, the set intersection problem~\cite{Chan2014},
has been considered.
However, the set intersection problem
can be reduced to the range mode enumeration problem, whereas the converse is not true.
We cannot use existing algorithms for the set intersection problem to solve the range
mode enumeration problem.  A simple modification of an existing algorithm~\cite{DM11}
works, but it takes $\Or(n^\epsilon)$ time for each output of an item
(see Theorem~\ref{thm-enum-ex2}).
We give faster solutions whose query time complexity is linear to the output size
plus some small term.  Table~\ref{table-enum} summarizes the results.

\begin{table}[bt]
	\centering
	\caption{Complexities of data structures for the range mode problem where $n$ is the number of terms of a string and 
	$m$ is the maximum frequency of an item.  The space complexities do not include one for the input string.}
	\label{table-rmo}
	\begin{tabular}{|c|c|c|c|}
		\hline
		Data structure & Space complexity (bits) & Query time complexity & conditions \\ \hline
	\cite{Petersen08} & $\Or\left({n^{2-2\ep}}\log n\right)$ & $\Or(n^{\ep})$ & $0\leq \ep \leq 1/2$\\	
	\cite{Chan2014}& $\Or\left({n^{2-2\ep}}\right)$ & $\Or(n^{\ep})$ & $0\leq \ep \leq 1/2$\\	
	\cite{PG09}& $\Or\left(\dis\frac{n^2\log\log n}{\log n}\right)$ & $\Or(1)$&\\	\cite{GreveJLT10}& $\Or(nm \log n)$ & $\Or(\log m)$ & \\
    \cite{DM11} & $\Or((n^{1-\ep} m+n)\log n)$ & $\Or(n^\ep+\log\log n)$ &$0\leq \ep \leq 1/2$\\ \hline
Theorem~\ref{thm-loglog2} & $\Or\left(4^knm\left(\frac{n}{m}\right)^{\frac{1}{2^{2^k}}}\right)$ & $\Or(2^k)$ & $k$ is any positive integer\\ 
Theorem~\ref{thm-m-small} & $\Or(nm)$ & $\Or(\min\{\log m, \log\log n\})$ &\\ 
Corollary~\ref{cor-loglog} & $\Or\left(nm\left(\log\log\frac{n}{m}\right)^2\right)$ & $\Or\left(\log\log\frac{n}{m}\right)$ & \\ \hline
	\end{tabular}
\end{table}

\begin{table}[bt]
	\centering
	\caption{Complexities of data structures for the range mode enumeration problem where $\ou$ denotes the number of solutions,
	$n$ the length of the input sequence, $m$ the maximum frequency of symbols, and $\ep$ is a parameter between $0$ and $1/2$ users can choose.
	Note that the space complexities does not contain that for the input sequence $S$.}
	\label{table-enum}
	\begin{tabular}{|c|c|c|}
		\hline
		Data structure & Space (bits) & Query time \\ \hline
Theorem~\ref{thm-enum-ex2} & $\Or(n^{2-2\ep}\log n)$ & $\Or(n^{\ep}\ou)$ \\ 
Theorem~\ref{thm-enum1} &$\Or\left(nm\left(\log\log\frac{n}{m}\right)^2+n\log n\right)$&$\Or\left(\log\log\frac{n}{m}+\ou \right)$ \\
Theorem~\ref{thm-enum2} & $\Or(nm+n\log n)$ & $\Or(\log m+\ou)$ \\ 
Theorem~\ref{thm-enum3} & $\Or(n^{1+\ep}\log n+n^{2-\ep})$ & $\Or(\log m+n^{1-\ep} +\ou)$\\ \hline
	\end{tabular}
\end{table}

The paper is organized as follows.
In Section~\ref{sec:pre}, we review basic properties of the range mode problem
and existing algorithms for the range mode problem.
We also explain fundamental data structures for storing integer sequences.
In Section~\ref{sec:rm}, we give our improved algorithms for the range mode problem.
In Section~\ref{sec-enum}, we give algorithms for the range mode enumeration problem.
Section~\ref{sec:conclusion} summaries the paper.
Some of the proofs, algorithms, and figures are given in the appendix.

\section{Preliminaries}\label{sec:pre}

\subsection{Basic properties}
To avoid confusion between modes and frequency of modes, from now on we consider range mode problems for 
not integer sequences but strings on an alphabet.  We define the following.
\begin{itemize}
\item $S$: input string
\item $n$: the length of string $S$
\item $\Sigma$: the set of characters (alphabet) of $S$
\item $f(l,r)$: the frequency of the modes of the substring $S[l,r]$
\item $m$: the frequency of a character with maximum frequency, that is, $m=f(0,n-1)$
\end{itemize}
We assume that there exists a bijection $\Sigma \rightarrow \{0,1,\ldots,|\Sigma|-1\}$ which can be computed 
in constant time.
We sometimes identify characters in the alphabet and integers.

\begin{lemma}
\label{lem-freq-require}[\cite{KMS05}]
If non-empty multisets $M,M_1$ and a multiset $M_2$ satisfies $M=M_1\cup M_2$
and if $x$ is a mode of $M$, at least one of the following holds.
\begin{itemize}
\item $x$ is a mode of $M_1$.
\item $x$ belongs to $M_2$.
\end{itemize}
\end{lemma}

\begin{lemma}
\label{lem-freq-first}
For $l_2<l_1\leq r_1<r_2$, if $f(l_1,r_1)=f(l_2,r_2)$, modes of range $[l_1,r_1]$
are also modes of range $[l_2,r_2]$.
\end{lemma}

\subsection{Algorithms for the range mode problem}
\label{subsec-dm}
We review the data structure with $\Or(n^{2-2 \epsilon})$-word space and $\Or(n^\epsilon)$ query time~\cite{DM11}.
The input string $S$ of length $n$ is partitioned into $n/s=n^{1-\ep}$ blocks of length $s=n^\ep$ each.
In addition to $S$, the data structure has the following four components.
\begin{description}
 \item[Two-dimensional array $A$]: For each character in the alphabet, an array for storing positions of its occurrences is used.
 \item[Array $B$]: For each position $i$ of $S$, $B[i]$ stores the number of times that the character $S[i]$ occurs 
        in the substring $S[0,\ldots,i-1]$.
 \item[Two-dimensional array $C$]: The $(i, j)$ entry of $C$ stores the frequency of modes of the substring from
        the $i$-th block to the $j$-th block.  That is, $C[i][j] = f(i\cdot s ,(j+1)\cdot s-1)$.
 \item[Two-dimensional array $D$]: The $(i, j)$ entry of $D$ stores one of the modes of the substring from
        the $i$-th block to the $j$-th block.
\end{description}
The space complexity is $O(n^{2-2\ep})$ words, for any fixed $0\leq \ep \leq 1/2$.
Using these arrays, any query $[l,r]$ is solved in $\Or(n^\ep)$ time as follows.
If a query is contained inside a block, we scan the range $[l,r]$ and
for each character in the alphabet, we count its number of occurrences.  This takes $\Or(s)=\Or(n^\ep)$ time.
If a query range $[l,r]$ lies on more than one block, we partition the query range
into prefix $[l,(b_l+1)s-1]$, span $[(b_l+1)s,b_rs-1]$, and suffix $[b_rs,r]$
where $b_l=\lfloor l/s\rfloor ,b_r=\lfloor r/s\rfloor $.  Note that the span may be empty.

From Lemma~\ref{lem-freq-require}, modes of range $[l,r]$ are either
(a) modes of the span, (b) a character in the prefix, or (c) a character in the suffix.
For (b) and (c), we scan the prefix and the suffix, and for each character in them,
we compute its frequency using the arrays $A$ and $B$ (for details refer to~\cite{DM11}).
For (a), one of the modes of the span and its frequency is obtained from $D[b_l][b_r]$
and $C[b_l][b_r]$, respectively.
This also takes $\Or(s)=\Or(n^\ep)$ time.

There exist improved data structures which are summarized in Table~\ref{table-rmo}.

\subsection{Representations of integer sequences}

We define $\ims(n,u)$ (increasing monotone sequences)
and $\dms(n,u)$ (decreasing monotone sequences)
as follows.
\begin{definition}
We define $\ims(n,u)$ as the set of all integer sequences 
$A$ of length $n$ such that
$0\leq A[0] \leq A[1] \leq \cdots \leq A[n-1] <u$.
We also define $\dms(n, u)$ as the set of all integer sequences
$A$ of length $n$ such that
$u> A[0] \geq A[1] \geq \cdots \geq A[n-1] \geq 0$.
\end{definition}

\begin{theorem}[\cite{Sumigawa2018}]
\label{thm-main}
For a sequence $A\in\ims(n,u)\ \ (n>u)$ and an integer $k \ge 0$, there exists a data structure
using $\Or(2^kn^{1/2^k}u^{1-1/2^k})$ bits which can compute
\begin{itemize}
\item $\ac(i,A)=A[i]$
\item $\bo(i,A)=\#\{j\mid A[j]>i\}$
\end{itemize}
in $\Or(2^k)$ time.
\end{theorem}


\begin{theorem}[FID~\cite{RRR07}]
For a bit-vector $B$ of length $n$ which contains $u$ ones, consider the following operations.
\begin{itemize}
\item $\ac(i,B)$: returns the $i$-th bit of $B$.
\item $\ra_c(i,B)$: returns $\#\{j\mid B[j]=c\}$.
\item $\sel_c(i,B)$: returns $\min \{j\mid \ra_c(j,B)=i\}$.
\end{itemize}
There exists a data structure which performs the operations in constant time
using $\log\binom{n}{u}+\Theta\left(\frac{n\log \log n}{\log n}\right)$ bits of space.
\end{theorem}

\section{Improved Data Structures for Range Mode Problem}\label{sec:rm}
We propose efficient data structures for the range mode problem
using $m$, the largest frequency of characters, as a parameter.


Consider the data structure of Section~\ref{subsec-dm} with $\ep=0$.
For simplicity we define $C[i][j]=1$ for any $i, j$ with $i>j$.
Then the $n \times n$ array $C$ satisfies the following property.
\begin{property}
\label{freq-table}
For any adjacent entries in the two-dimensional array $C$, it holds
\begin{align*}
C[i][j]\leq C[i][j+1]\leq C[i][j]+1\ \ \  &(0\leq i<n,\ 0\leq j<n-1),\\
C[i][j]\leq C[i-1][j]\leq C[i][j]+1\ \ \  &(1\leq i<n,\ 0\leq j<n).
\end{align*}
\end{property}
From the definition, $C$ also satisfies:
\begin{property}
Any entry of $C$ is an integer between $1$ and $m$.
\end{property}
Below we propose a data structure for storing $C$ in a compressed form and supporting constant time access.

\subsection{An efficient representation of the array $C$}
We define the set of two-dimensional arrays which have both column-wise and row-wise monotonicity as follows.
\begin{definition}
We define the set of two-dimensional arrays $A[0\ldots n-1][0\ldots n-1]$
which satisfy all the following inequalities as $\tims (n,m)$.
\begin{align*}
A[i][j]\leq A[i][j+1]\ \ \  &(0\leq i<n,\ 0\leq j<n-1),\\
A[i][j]\leq A[i+1][j]\ \ \  &(1\leq i<n,\ 0\leq j<n),\\
0\leq A[i][j] < m\ \ \  &(0\leq i<n,\ 0\leq j<n).	
\end{align*}
\end{definition}

\begin{theorem}
\label{thm-2d-main}
Let $A$ be a two-dimensional array in $\tims(n,m)$ ($n\geq m$) and $k$ be a non-negative integer.
There exists a data structure $S_k(n,m)$ which can output an entry of $A$ in
$\Or\left(2^k\right)$ time using $\Or\left(4^knm\left(\frac{n}{m}\right)^{\frac{1}{2^{2^k}}}\right)$ bits of space.
\end{theorem}
\begin{proof}
We prove by induction on $k$ that
there exists a data structure $S_k(n,m)$ using at most
$c4^knm\left(\frac{n}{m}\right)^{\frac{1}{2^{2^k}}}$ bits of space, where
$c$ is some constant satisfying:
\begin{itemize}
\item There exists a data structure $Z(n,m)$ using at most ${c^{1/3}\sqrt{nm}}$ bits of space
which can read an entry of $\ims(n,m)$ in constant time.
\end{itemize}
Below we show such a constant $c$ exists, if we use the data structure of Theorem~\ref{thm-main}.

For $k=0$, we use the data structure $Z(n,m)$ of Theorem~\ref{thm-main} for storing each column.
Then the space usage is at most $c^{1/3}n^{3/2}m^{1/2}$ bits, which is at most $cn^{3/2}m^{1/2}$ bits
and the claim holds.

Now we assume that the data structure $S_{k-1}$ exists, and prove $S_k$ also exists.
We partition the two-dimensional array into $u\times u$ blocks where $u=n/t$, each of which has
$t=\left(\frac{n}{m}\right)^{\frac{1}{2^{2^{k-1}}}}$ columns and rows.
The block corresponding to $A[it,\ldots,(i+1)t-1][jt,\ldots,(j+1)t-1]$ is a $t\times t$ two-dimensional array
and denoted by $B_{i,j}$.
We define flatness of a block as follows.
\begin{definition}
A block is called \emph{flat} if all the entries in the block are identical.
\end{definition}
We also define the height of a block.
\begin{definition}
The \emph{height} of block $B_{i,j}$, denoted by $d_{i,j}$, is defined as
\begin{align*}
d_{i,j}&=B_{i,j}[t-1][t-1]-B_{i,j}[0][0]+1\\
(&=A[(i+1)t-1][(j+1)t-1]-A[it][jt]+1).
\end{align*}
That is, the height of a block is the difference between
the maximum and the minimum values in the block, plus one.
\end{definition}
We prove the following:
\begin{theorem}
\label{non-flat-2d}
Among $u^2$ blocks, there are at most $2um$ non-flat blocks.
\end{theorem}
To prove it, 
we define $k$-th \emph{boundary} in a block for $k=0,1,\ldots,m-1$ as follows.
\begin{definition}
\label{def-boundary}
For a two-dimensional array $A\in \tims(n,m)$, consider the $(n+1)\times(n+1)$ grid graph $G$.
The $k$-th boundary of $A$ is defined as the edge set of $G$ satisfying:
\begin{align*}
&\{((i,j),(i+1,j))|A[i][j-1]<k \mbox{ and } A[i][j]\geq k\}\\
&\ \ \ \ \ \ \cup \{((i,j),(i,j+1))|A[i-1][j]<k \mbox{ and } A[i][j]\geq k\},
\end{align*}
where we assume $A[-1][\cdot]=A[\cdot][-1]=-1$.
\end{definition}

Then the following holds.
\begin{property}
\label{prop_bound_path}
The $k$-th boundary is a shortest path from vertex $(0,n)$ to vertex $(n,0)$ of the grid graph $G$.
That is, if we regard the path as a directed path from $(0,n)$ to $(n,0)$,
the edges in the path are of the form of either $(i,j)\rightarrow(i+1,j)$ or $(i,j)\rightarrow(i,j-1)$.
\end{property}

\begin{example}
Figure~\ref{fig-ims2-ex} shows an $\tims(4,4)$ array and its second boundary.
	\begin{figure}[t]
		\begin{center}
			\includegraphics[width=40mm, angle=90]{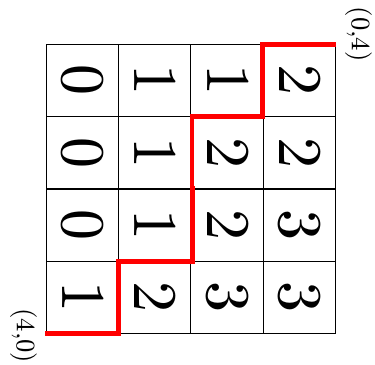}
			\end{center}
		\caption{An $\tims(4,4)$ array.  The second boundary of its grid graph is shown in a red bold line.
		We can see that the boundary is a shortest path from vertex $(0,4)$ to vertex $(4,0)$ of the grid graph.}
		\label{fig-ims2-ex}
	\end{figure}
\end{example}

\begin{proof}[Proof of Theorem~\ref{non-flat-2d}]
It is equivalent that a block is flat, and that any boundary does not pass inside the block.
For each of $m$ boundaries, the number of blocks in which the boundary passes is $2u$.
Therefore the number of blocks which contains at least one boundary in it is at most $2um$.
\end{proof}

Based on this property, we use the following data structure.
\begin{description}
 \item[$E$: to store $\tims(u,m)$]\mbox{}\\
We define $E[i][j]=A[it][jt]\ (0\leq i <u, 0\leq j <u)$.  It is clear that $E\in \tims(u,m)$.
\item[$F_{i,j}$: to store differences inside block $B_{i,j}$]\mbox{}\\
For non-flat block $B_{i,j}$, we define $F_{i,j}[x][y]:=B_{i,j}[x][y]-B_{i,j}[0][0]$.  Then it holds $F_{i,j}\in \tims(t,d_{i,j})$.
\end{description}
It holds for the original array $A$, $A[i][j]=E[i/t][j/t]+F_{i/t,j/t}[i\%t][j\%t]$,
and for flat block $B_{i,j}$, the two-dimensional array $F_{i,j}$ is the zero-value array.
Then an access to the array $A$ is done by Algorithm~\ref{algo-2d-rec}.
At line~\ref{if-2d-flat} of the algorithm, it is necessary to decide if a block is flat or not.
Because a naive data structure using a $u\times u$ Boolean array is space-consuming,
we develop a space-efficient solution. To do so, we define the following mapping.

\begin{definition}[Mapping to decide if a block is flat or not]
\label{def-phi}
We define a mapping from a block number to a pair of integers
$\Phi:\ \{0,\ldots ,u-1\}^2 \rightarrow \{0,\ldots,m-1\}\times \{0,\ldots,2u-2\}$
as
\begin{align*}
(i,j) \mapsto (E[i][j],i-j+u-1).
\end{align*}
\end{definition}
We obtain the following.
\begin{theorem}
For any two distinct non-flat blocks $B_{i_1,j_1}$ and $B_{i_2,j_2}$, it holds $\Phi(i_1,j_1)\neq\Phi(i_2,j_2)$.
\end{theorem}
\begin{proof}
Let $B_{i_1,j_1}$ and $B_{i_2,j_2}$ be two distinct non-flat blocks.
From the definition of $\Phi$ the claim holds if $E[i_1][j_1]\neq E[i_2][j_2]$.
If $E[i_1][j_1]= E[i_2][j_2]$, the $E[i_2][j_2]$-th boundary must pass both $B_{i_1,j_1}$ and $B_{i_2,j_2}$.
It is however not possible to pass both of them if $i_1-j_1=i_2-j_2$ from Property~\ref{prop_bound_path}.
Thus it holds $\Phi(i_1,j_1)\neq\Phi(i_2,j_2)$.
\end{proof}
We also define a mapping, which is something like an inverse of $\Phi$.
\begin{definition}
We define a mapping
\begin{align*}
\Psi:\  \{0,\ldots,m-1\}\times \{0,\ldots,2u-2\} \rightarrow \{0,\ldots ,u-1\}^2\cup \{\bot\}
\end{align*}
as $\Psi(x,y)=(i,j)$ if there exists a non-flat block $B_{i,j}$ with $\Phi(i,j)=(x,y)$,
and $\Psi(x,y)=\bot$ otherwise.
\end{definition}
Then it holds
$\mbox{block } b_{i,j} \mbox{ is not flat } \Leftrightarrow \Psi(\Phi(i,j))=(i,j)$.
To use this fact, we have to compute both $\Psi$ and $\Phi$.
We can compute $\Phi$ at line 3 of Algorithm~\ref{algo-2d-rec} from Definition~\ref{def-phi}.
To compute $\Psi$, we use the following.
\begin{lemma}
\label{lem-2d-mono}
Assume that $\Psi(x,y_1)=(i_1,j_1), \Psi(x,y_2)=(i_2,j_2)$.
Then if $y_1\leq y_2$, it holds $i_1\leq i_2$ and $j_1\geq j_2$.
\end{lemma}
Finally we obtain:
\begin{theorem}\label{th:psi}
$\Psi(x,y)$ is computed in constant time using a data structure of
$(2\sqrt{2}c^{1/3}+2)\dis\frac{nm}{t}$ bits.
\end{theorem}

Therefore the decision at line~\ref{if-2d-flat} of Algorithm~\ref{algo-2d-rec} is done
by using Algorithm~\ref{algo-flat}.

We analyze the space complexity of the data structure $S_k$.
\begin{lemma}\label{lem:sk}
The size of the data structure $S_k$ is
$c4^knm\left(\frac{n}{m}\right)^{\frac{1}{2^{2^{k}}}}$
\end{lemma}

Next we consider the time $T_k$ to access an entry of $S_k$.
In Algorithm~\ref{algo-2d-rec}, lines 3 and 5 take $T_{k-1}$ time.
For other lines including the call to Algorithm~\ref{algo-flat} it takes constant time.  Therefore it holds
$T_k=2T_{k-1}+\Or(1)$
and we obtain $T_k=\Or(2^k)$.

This completes the proof of Theorem~\ref{thm-2d-main}.
\end{proof}
We also obtain:
\begin{theorem}
\label{thm-2d-access}
There exists a data structure of $\Or\left(nm\left(\log\log\frac{n}{m}\right)^2\right)$ bits
supporting an access to $\tims(n,m)\ (n>m)$
in $\Or\left(\log\log\frac{n}{m}\right)$ time.
\end{theorem}
\begin{proof}
We obtain the claim by letting $k = \log \log \log \frac{n}{m}$ in Theorem~\ref{thm-2d-main}.
\end{proof}
By using this data structure for a two-dimensional array $C$ satisfying Property~\ref{freq-table},
we can compute the frequency $f$ of the modes of a query range $[l,r]$.
From Lemma~\ref{lem-freq-first}, a mode is obtained by computing $S[\min \{x\mid C[l][x]=f\}]$.
To compute this, consider the following data structure.

\begin{theorem}
\label{thm-2d-pred}
There exists a data structure of $\Or(nm)$ bits for given column number $r$ of $A\in \tims(n,m)$
and a value $h$, to compute $\min\{x\mid A[r][x]\geq h\}$ in constant time.
\end{theorem}
\begin{proof}
For two-dimensional array $A$, consider the boundaries of Definition~\ref{def-boundary}.
We can say that $\min\{x\mid A[r][x]\geq h\}$ is the minimum row number of elements in $r$-th column
which are above the $h$-th boundary.
Recall that boundaries are shortest paths in the grid graph from vertex $(0,n)$ to vertex $(n,0)$.
Consider to encode the $m$ boundaries as follows.
\begin{definition}[A bit-vector representation of a boundary]
\label{bound-bit}
We encode a boundary by a bit-vector of $2n$ bits as follows.
Initially the bit-vector is set empty.
We traverse the graph from vertex $(0,n)$ to vertex $(n,0)$ along the boundary,
and append $0$ when we go down, and $1$ when we go right, to the end of the bit-vector.
\end{definition}
The bit-vector for a boundary has $n$ zeros and $n$ ones.  There are $m$ such bit-vectors
and the space complexity is $\Or(nm)$ bits in total.
Let $B_k$ denote the bit-vector for the $k$-th boundary.  From definition, it holds
$\min\{x\mid A[r][x]\geq h\}=n-\ra_0(\sel_1(r,B_k),B_k)$.
This can be computed in constant time by using FID.
\end{proof}

\begin{theorem}
\label{thm-loglog2}
In addition to the string $S$, by using a data structure of
$\Or\left(4^knm\left(\frac{n}{m}\right)^{\frac{1}{2^{2^k}}}\right)$ bits,
we can solve the range mode problem in $\Or\left(2^k\right)$ time.
\end{theorem}

\begin{proof}
We can use Algorithm~\ref{algo-loglog}, where the two-dimensional array $C$ satisfies the following:
\begin{align*}
C[i][j]= \begin{cases}
    f(i,j) & (i\leq j), \\
    1 & ({\rm otherwise}).
  \end{cases}
\end{align*}
By storing the rows or columns in reverse order and subtracting one from all values,
$C$ belongs to $\tims(n,m)$.

In Algorithm~\ref{algo-loglog}, the data structures of Theorems~\ref{thm-2d-main} and~\ref{thm-2d-pred} are used.
The space complexity includes $\Or\left(4^knm\left(\frac{n}{m}\right)^{\frac{1}{2^{2^k}}}\right)$ bits
for Theorem~\ref{thm-2d-main} and $\Or(nm)$ bits for Theorem~\ref{thm-2d-pred}, and therefore the total
space complexity is $\Or\left(4^knm\left(\frac{n}{m}\right)^{\frac{1}{2^{2^k}}}\right)$ bits.
\end{proof}

By letting $k = \log \log \log \frac{n}{m}$ in Theorem~\ref{thm-loglog2}, we obtain:
\begin{corollary}
\label{cor-loglog}
In addition to the string $S$, using a data structure of $\Or\left(nm\left(\log\log\frac{n}{m}\right)^2\right)$ bits,
the range mode problem is solved in $\Or\left(\log\log\frac{n}{m}\right)$ time.
\end{corollary}
This data structure is superior to the data structure of \cite{DM11} with $\ep=0$,
which has space complexity $\Or(nm\log n)$ bits and query time complexity $\Or(\log\log n)$,
in terms of both time and space.

\subsection{Efficient data structure for small $m$}
Instead of using the two-dimensional array $C$ storing frequencies of all the ranges,
we can compute the frequency of modes using only the bit-vector representation of the boundaries of Definition~\ref{bound-bit}.
\begin{theorem}
\label{thm-dec}
For a two-dimensional array $A\in \tims(n,m)$, there exists a data structure of $\Or(nm)$ bits that given
$i,j,k$, to decide if $A[i][j]\geq k$ in constant time.
\end{theorem}
\begin{proof}
We store all the bit-vectors $B_0,\ldots,B_{m-1}$ representing the boundaries.
It is enough to decide if the $k$-th boundary of $A$ is either in the $(0,0)$'s side or $(n,n)$' size with respect to $(i,j)$.
This is done by Algorithm~\ref{algo-bou}, which runs in constant time.
\end{proof}

From Theorem~\ref{thm-dec}, we can compute $C[i][j] = \max\{k|C[i][j]\geq k\}$ in $\Or(\log m)$ time by a binary search on $k$.
\begin{theorem}
\label{thm-2d-pred2}
There exists a data structure of $\Or(nm)$ bits that for a given $(r,c)$ of $A\in \tims(n,m)$, computes the value $h$ such that $h-\log n < A[r][c]\leq h$ in $O(\log\log n)$ time.
\end{theorem}
\begin{proof}
For every column $r$ of $A$, we take the set of (at most $m/\log n$) row indexes $c$ such that $A[r][c] > A[r][c-1]$ and $A[r][c]$ is a multiple of $\log n$, and store the set using a predecessor data structure~\cite{Willard83}. The space usage for each column is $O((m/\log n)\log n) = O(m)$ bits, and hence overall $O(nm)$ bits to represent $A$. The query is supported by finding the predecessor of $c$
in the predecessor data structure corresponding to the column $r$.
\end{proof}
Now using the structure of Theorem~\ref{thm-dec}, we can compute $C[i][j] = \max\{k|C[i][j]\geq k\}$ in $\Or(\log \log n)$ time by a binary search on $k$, after narrowing down the length of the range of $k$ to $O(\log n)$ using the structure of Theorem~\ref{thm-2d-pred2}.
Furthermore, from Theorem~\ref{thm-2d-pred}, we can compute an index for modes in constant time.
Now we obtain the following theorem.

\begin{theorem}
\label{thm-m-small}
In addition to the input string $S$, by using a data structure of $\Or(nm)$ bits,
the range mode problem is solved in $\Or(\min\{\log m, \log\log n\})$ time.
\end{theorem}
Table~\ref{table-rmo} summarizes the proposed and known data structures.

\section{Range Mode Enumeration Problem}
\label{sec-enum}

Below we consider range modes of a string $S$ with alphabet size $\sigma$ instead of a sequence $A$.
We evaluate algorithms with their space complexity and query time complexity
using the size of the output $\ou$ as a parameter.

\subsection{Algorithms using existing data structures}
Data structures for the range mode problem return only arbitrary one item among all range modes.
Instead here we consider a data structure for the problem which returns the leftmost index and the frequency
of range modes, where the leftmost index is defined as follows.
\begin{definition}
For a string $S$ and query range $[l,r]$, the leftmost index of range modes is defined as
$\min\{x\mid S[x] \mbox{ is an item with the largest frequency in the query range } [l,r]\}$.
\end{definition}

\begin{lemma}
\label{lem-enum-ex}
Let $D$ be a data structure which returns the leftmost index of range modes for
a query range $[l,r]$
in time $t$ using $s$ space,
there exists a data structure which solves the range mode enumeration problem
in time $(t+O(1)) \ou$ using $s$ space.
\end{lemma}
\begin{proof}
Algorithm~\ref{algo-enum-left} solves the problem using the data structure $D$.
The algorithm narrows the query range gradually.  Because the data structure $D$ returns the leftmost index $i$ of
range modes, the number of range modes for the new query range $[i+1, r]$ is exactly one smaller than that of
the current query range.  Therefore the while loop at line~\ref{algo-whi} of the algorithm is executed $\ou+1$ times,
and the total time complexity is $(t+\Or(1)) \ou$.
\end{proof}

\begin{lemma}
\label{lem-ex-enu}
There exists a data structure for finding the leftmost index of range modes and their frequency
in time $\Or(n^\ep)$ using a data structure with space complexity $\Or(n^{2-2\ep})$ words
in addition to the input string $S$.
\end{lemma}
\begin{proof}
We slightly change the data structure of~\cite{DM11} described in Section~\ref{subsec-dm}.
Instead of the two-dimensional array $D$ storing modes of block ranges,
we create another two-dimensional array $D'$ storing leftmost indices of block ranges.
Then we can find the leftmost index of span in constant time.
For items appearing in the prefix and the suffix,
we can find the leftmost index and its frequency using the same algorithm.
Algorithm~\ref{algo-left-ind} gives a pseudo code.
\end{proof}

From Lemmas~\ref{lem-enum-ex} and~\ref{lem-ex-enu}, we obtain the following.
\begin{theorem}
\label{thm-enum-ex2}
There exists a data structure for the range mode enumeration problem
solving a query in $\Or(n^{\ep}\ou)$ time using $\Or(n^{2-2\ep} \log n)$ bits.
\end{theorem}

\subsection{More efficient data structures for enumeration}
\begin{definition}
The \emph{mode index set} for a query range $[l,r]$ of the range mode enumeration problem
is the set of the rightmost position of each mode in the query range.  That is,
\begin{align*}
\{i\mid S[i] \mbox{ is a mode and } S[i]\neq S[j] \mbox{ for any } j=i+1,i+2,\ldots,l \}.
\end{align*}
\end{definition}
Because the set of all range modes can be obtained from the mode index set,
below we focus on finding the mode index set.

Define $n$ bit-vectors $B[0],\ldots,B[n-1]$ of length $n$ each as follows.
\begin{align*}
B[i][j]=1 \Leftrightarrow i\leq j \mbox{and $S[j]$ is a mode of range $[i,j]$}
\end{align*}

Using these bit-vectors, we obtain:
\begin{theorem}
The set of modes for a query range $[l,r]$ is 
$\{x\mid f(l,x)=f(l,r) \wedge B[l][x]=1\}$.
\end{theorem}
\begin{proof}
From the definition of $B[l][x]$, $S[x]$ is a mode of range $[l,r]$.
From Lemma~\ref{lem-freq-first}, $S[x]$ is also a mode of range $[l,r]$.
Conversely, for any index $x$ contained in the index set for range $[l,r]$, 
it holds $f(l,x)=f(l,r)$ and $B[l][x]=1$.  Therefore these two sets coincide.
\end{proof}
Therefore we can enumerate items in the mode index set by using Algorithm~\ref{algo-n2bits}.

Consider complexities of the algorithm.
For the data structure $B$, which consists of $n$ bit-vectors of length $n$, we use $\Or(n^2)$ bits.
We also use $\Or(n^2)$ bits for the array $C$ storing frequencies using bit-vectors,
which is used to obtain frequency of modes of a query range.
Therefore the total space is $\Or(n^2)$ bits.
As for the time complexity, the algorithm executes Line~\ref{algo-n2while} for $\Or(\ou)$ times.
Lines~\ref{algo-n21} and~\ref{algo-n22} takes constant time if we use data structures for bit-vectors.
Therefore the total time complexity is $\Or(\ou)$.  We obtain the following basic data structure.
\begin{theorem}
\label{thm-n2bits}
There exists a data structure for the range mode enumeration problem
which computes the mode index set in $\Or(\ou)$ time using a data structure of $\Or(n^{2})$ bit space
in addition to the input string $S$.
\end{theorem}

Now we improve this using a parameter $m$, the frequency of modes of the entire range.
The following lemma holds for the two-dimensional bit-array $B$.
\begin{lemma}
There is the following relation between function $f$ and bit-array $B$.
\begin{align*}
\mbox{If } f(i,j)=f(i+1,j) \mbox{ and } B[i+1][j]=1, \mbox{ then } B[i][j]=1.
\end{align*}
\end{lemma}
\begin{proof}
Because $B[i+1][j]=1$, $S[j]$ is a mode of range $[i+1,j]$.
Using Lemma~\ref{lem-freq-first}, it holds $S[j]$ is also a mode of range $[i,j]$.
From the definition of $B$, we obtain $B[i][j]=1$.
\end{proof}
Using this property, we define $m$ integer sequences $H[1],\ldots,H[m]$ of length $n$ each.
\begin{definition}
\label{def-h}
Define integer sequences $H[1],\ldots,H[m]$ as follows.
\begin{align*}
H[i][j]=\max\left\{ \{k\mid f(k,j)=i \mbox{ and } B[k][j]=1\} \cup \{-1\} \right\}.
\end{align*}
\end{definition}
The bit-array $B$ and the sequences $H$ have the following relation.
\begin{lemma}
$B[i][j]=1\Leftrightarrow H[f(i,j)][j]\geq i$.
\end{lemma}
\begin{example}
Figure~\ref{fig-nmbit} shows an example of bit-array $B$ and sequences $H$ for string $S=\mbox{``abcbfcdaacfbcgba''}$.
\end{example}

Algorithm~\ref{algo-n2bits} enumerates indices with bits being set in range $[b, r]$ of bit-vector $B[l]$.
Here for any $t$ with $b\leq t\leq r$, the value of $f(l,t)$ is always $g$.
Therefore this operation is identical to enumerate all indices in range $[b, r]$ of sequence $H[g]$
whose value is at least $l$.  This problem can be regarded as the range maximum problem.
\begin{definition}[Range Maximum Problem (RMQ)]
Given a sequence $A$ and query range $[l,r]$, the range maximum problem
asks an index of the maximum value in the sub-sequence $A[l,\ldots,r]$.
\end{definition}
\begin{theorem}[\cite{Fischer2011}]
For the range maximum problem of size $n$, there exists a data structure with space complexity $2n+\order(n)$ bits
and query time complexity $\Or(1)$.
\end{theorem}
\begin{theorem}\label{th:rmq}
Consider the following problem: Given a sequence $A$, a query range $[l,r]$ and a threshold $t$,
compute $\{l\leq k\leq r\mid A[k]\geq t\}$.
If there exists an oracle to check if it holds $A[k]\geq t$ for some $k$ in constant time,
there exists a data structure for the problem with $\Or(n)$ bits and $\Or(\ou)$ query time.
\end{theorem}
Consider a data structure to decide if $H[f(l,r)][k]\geq l$ or not for finding the index set.
This can be done by using the arrays $A, B$ in Section~\ref{subsec-dm} because
it is equivalent that $H[f(l,r)=g][k]\geq l$ and the frequency of $S[k]$ in range $[l,r]$ is at least $g$.

From the observation above, it is enough to use the following data structures to enumerate the solutions.
\begin{description}
 \item[Two-dimensional array $A$ storing positions of occurrences of symbols]\mbox{}\\ 
        An array to store positions of occurrences in ascending order for each symbol of the alphabet
 \item[Array $B$ to store ranks for strings]\mbox{}\\
        An array storing the rank for each index of $S$, that is, $B[i]$ stores the number of times that
        the symbol $S[i]$ appears in the substring $S[0,\ldots,i-1]$.
 \item[Two-dimensional array $C$ storing frequencies of modes for all ranges]\mbox{}\\
        The $(i,j)$ entry of the array $C$ stores the frequency of the modes for range $[i,j]$.
\item[$m$ bit-vectors $D$ storing boundaries of the array $C$]\mbox{}\\
      The array stores $m$ bit-vectors of Definition~\ref{bound-bit}.
\item[Two-dimensional array $H$ storing $m$ RMQ data structures for arrays of length $n$ each]\mbox{}\\
      An array storing $m$ sequences of Definition~\ref{def-h} as RMQ data structures.
      The sequences themselves are not stored.
\end{description}
The space complexity of the two-dimensional array $C$ varies depending on which data structure is used.
For example, we can use ones in Theorems~\ref{thm-2d-access} and~\ref{thm-m-small}.
The space complexities of $A, B, D, H$ are $\Or(n\log n)$ bits, $\Or(n\log n)$ bits, $\Or(nm)$ bits, $\Or(nm)$ bits,
respectively.

The pseudo code is given in Algorithm~\ref{algo-enum}.
Only Line~\ref{algo-en21} cannot be done in constant time.  For other lines,
the time complexity is proportional to the number of times the function ${\it range}$ is executed,
and it is $\Or(\ou)$.

To recap, the complexities of the algorithms become $\Or(S+nm+n\log n)$ bit space and $\Or(T+\ou)$ query time,
where $S$ is the space complexity of the two-dimensional array $C$, and $T$ is the time complexity to access an entry of $C$.
Using Theorems~\ref{thm-2d-access} and~\ref{thm-m-small}, we obtain the following.

\begin{theorem}
\label{thm-enum1}
There exists a data structure with space complexity $\Or\left(nm\left(\log\log\frac{n}{m}\right)^2+n\log n\right)$ bits
in addition to the input string $S$, which solves a query in time $\Or\left(\log\log\frac{n}{m}+\ou \right)$.
\end{theorem}
\begin{theorem}
\label{thm-enum2}
There exists a data structure with space complexity $\Or(nm+n\log n)$ bits
in addition to the input string $S$, which solves a query in time $\Or(\log m+\ou)$.
\end{theorem}

By combining the data structure of~\cite{DM11}, we can further reduce the space complexity.
Consider a string $S_1$ which stores symbols of $S$ whose frequencies are at least $n^{1-\ep}$,
and a string $S_2$ which stores the rest of the symbols.
The string $S_1$ stores at most $n^\ep$ distinct symbols.
Using the data structures of~\cite{DM11} and Theorem ~\ref{thm-enum2}
for $S_1$ and $S_2$ respectively, the following holds.
\begin{theorem}
\label{thm-enum3}
There exists a data structure with space complexity $\Or(n^{1+\ep}\log n+n^{2-\ep})$ bits
in addition to the input string $S$, which solves a query in time $\Or(\log m+n^{1-\ep} +\ou)$.
\end{theorem}

The proposed data structures for the range mode enumeration problem are summarized in Table~\ref{table-enum}.

\section{Concluding Remarks}\label{sec:conclusion}
In this paper, we have given more efficient algorithms
for the indexing version of the range mode problem.
Our algorithms are more space- and time- efficient for small maximum frequency case than existing ones.
We have also considered a natural extension of the range mode problem:
range mode enumeration problem and given fast algorithms.

There are other related problems like Boolean matrix multiplication
problem.  A future work, we will give efficient algorithms for these
problems.




%


\bibliography{dfs}

\appendix

\section{Omitted Proofs}

\subsection{Proof of Lemma~\ref{lem-freq-require}}
\begin{proof}
We prove by contradiction.
Let $x$ be a mode of $M$, and assume $x$ is not a mode of $M_1$ and
$x\notin M_2$.  Let $y\in M_1$ be a mode of $M_1$.
From the definition the multiplicity of $y$ in $M_1$ is strictly larger than that of $x$ in $M_1$.
Because $x\notin M_2$, the multiplicity of $x$ in $M$ is equal to that of $x$ in $M_1$,
and it is smaller than that of $y$ in $M$.
This contradicts that $x$ is a mode of $M$.
\end{proof}

\subsection{Proof of Lemma~\ref{lem-freq-first}}
\begin{proof}
Let $c$ be any mode in range $[l_1,r_1]$ and $m$ be its frequency. 
Because range $[l_2,r_2]$ contains range $[l_1,r_1]$,
the frequency of $c$ in range $[l_2,r_2]$ is at least $m$.
On the other hand, because $f(l_1,r_1)=f(l_2,r_2)=m$,
the frequency of $c$ in range $[l_2,r_2]$ is at most $m$.
Therefore the frequency of $c$ in range $[l_2,r_2]$ becomes $m$
and $c$ is also a mode in range $[l_2,r_2]$.
\end{proof}

\subsection{Proof of Lemma~\ref{lem:sk}}

If $u>m$, the space complexity of the two-dimensional array $E$ is, from the assumption of $S_{k-1}(u,m)$,
\begin{align*}
	c4^{k-1}um\left(\frac{u}{m}\right)^\frac{1}{2^{2^{k-1}}}
&=c4^{k-1}nm\left(\frac{m}{n}\right)^\frac{1}{2^{2^{k-1}}}\left(\frac{u}{m}\right)^\frac{1}{2^{2^{k-1}}}\\
&=c4^{k-1}nm\left(\frac{1}{t}\right)^\frac{1}{2^{2^{k-1}}}\\
&=c4^{k-1}nm\left(\frac{m}{n}\right)^{\frac{1}{2^{2^{k-1}}}\frac{1}{2^{2^{k-1}}}}\\
&\leq c4^{k-1}nm.
\end{align*}
If $u\leq m$, it can be stored in $c^{1/3}u^{3/2}m^{1/2}$ bits by using the data structure $Z(u,m)$ for each row.
Therefore for any case
$E$ can be stored in at most $c4^{k-1}nm$ bits.

Next we consider the space complexity of storing differences inside non-flat blocks.
\begin{lemma}
\label{lem-d-sum}
For the summation of all $d_{i,j}$, it holds
$\dis\sum_{ 0\leq i<u\atop0\leq j<u}d_{i,j}\leq \frac{2mn}{t}$.
\end{lemma}
\begin{proof}
From the column-wise and row-wise monotonicity, for each $l=-u+1,\ldots ,u-1$, it holds
$\dis\sum_{i-j=l}d_{i,j}\leq m$.
By summing this for all $l$, we obtain the claim. 
\end{proof}

Consider the space complexity of the data structure storing $F_{i,j}\in \tims(t,d_{i,j})$ for non-flat blocks $B_{i,j}$.
If $t>d_{i,j}$, by using $S_{k-1}(t,d_{i,j})$, the space becomes $c4^{k-1}td_{i,j}\left(\frac{t}{d_{i,j}}\right)^{\frac{1}{2^{(2^{k-1})}}}$ bits.
If $t\leq d_{i,j}$, we store each row of the two-dimensional array in $t\cdot c^{1/3} \sqrt{td_{i,j}}$ bits
by using $Z(t,d_{i,j})$ which support constant access.
For both time and space, the former case has worse complexities.
Therefore we analyze the space by assuming every block is stored in
$c4^{k-1}td_{i,j}\left(\frac{t}{d_{i,j}}\right)^{\frac{1}{2^{(2^{k-1})}}}$ bits.
\begin{align*}
\dis\sum_{i,j\atop B_{i,j} \mbox{ not flat}} c4^{k-1}td_{i,j}\left(\frac{t}{d_{i,j}}\right)^{\frac{1}{2^{2^{k-1}}}}&\leq \dis\sum_{i,j}c4^{k-1}td_{i,j}\left(\frac{t}{d_{i,j}}\right)^{\frac{1}{2^{2^{k-1}}}}
\leq \dis\sum_{i,j}c4^{k-1}td_{i,j}t^{\frac{1}{2^{2^{k-1}}}}\\
&\leq  c4^{k-1} t^{1+\frac{1}{2^{2^{k-1}}}}\dis\sum_{i,j}d_{i,j}
\leq  c4^{k-1} t^{1+\frac{1}{2^{2^{k-1}}}}\cdot\frac{2nm}{t}\\
&= 2c\cdot 4^{k-1}nm\left(\frac{n}{m}\right)^{\frac{1}{2^{2^{k-1}}}\frac{1}{2^{2^{k-1}}}}
= 2c\cdot 4^{k-1}nm\left(\frac{n}{m}\right)^{\frac{1}{2^{2^{k}}}}.
\end{align*}

We also need to store pointers to the data structures $F_{i,j}$ because their size varies depending on $(i,j)$.
As a bijection between $\{0,\ldots,m-1\}\times \{0,\ldots,2u-2\}$ and $\{0,1,\ldots,m(2u-1)\}$, we define
$(i,j)\mapsto i(2u-1)+j$.
By using this, we can regard the pointers to the data structures
as a monotone increasing sequence $P$ with $2um$ terms and range $2c\cdot 4^{k-1}nm\left(\frac{n}{m}\right)^{\frac{1}{2^{2^{k}}}}$.
By representing $P$ by the data structure $Z(2um,2c\cdot 4^{k-1}nm\left(\frac{n}{m}\right)^{\frac{1}{2^{2^{k}}}})$, 
it holds
\begin{align*}
c^{1/3}\sqrt{ 2um \cdot c\cdot 4^{k-1}nm\left(\frac{n}{m}\right)^{\frac{1}{2^{2^{k}}}}}\leq c^{5/6}2^{k}nm.
\end{align*}
Therefore the space is upper-bounded by $c^{5/6}2^{k}nm$ bits.


The total space of the data structures for $S_k$ is:
\begin{align*}
\underbrace{c4^{k-1}nm}_{\mbox{array }E}+\underbrace{(c^{1/3}\cdot2\sqrt{2}+2)\dis\frac{nm}{t}}_{{\Psi}}+\underbrace{2c\cdot 4^{k-1}nm\left(\frac{n}{m}\right)^{\frac{1}{2^{2^{k}}}}}_{\mbox{total space for }F}+\underbrace{c^{5/6}2^{k}nm}_{P}\ {\rm bits}.
\end{align*}
By letting $c\geq 10^6$, for any positive integer $k$, it holds
\begin{align*}
&\ \ \ \ \ c4^{k-1}nm+(2\sqrt{2}c^{1/3}+2)\dis\frac{nm}{t}+2c\cdot 4^{k-1}nm\left(\frac{n}{m}\right)^{\frac{1}{2^{2^{k}}}}+c^{5/6}2^{k}nm\\
&\leq (c4^{k-1}+c^{1/3}2\sqrt{2}+2+2c\cdot 4^{k-1}+c^{5/6}2^{k})nm\left(\frac{n}{m}\right)^{\frac{1}{2^{2^{k}}}}\\
&\leq c4^knm\left(\frac{n}{m}\right)^{\frac{1}{2^{2^{k}}}}.
\end{align*}
This proves there exists a data structure of 
$c4^knm\left(\frac{n}{m}\right)^{\frac{1}{2^{2^{k}}}}$ bits
for $S_k(n,m)$.

\subsection{Proof of Theorem~\ref{th:psi}}
\begin{proof}
We use a two-dimensional Boolean array $K$ of $2um$ bits storing
for each member $(x,y)$ of $\{0,\ldots,m-1\}\times \{0,\ldots,2u-2\}$,
True if $\Psi(x,y)\neq\bot$ and False if $\Psi(x,y)=\bot$.
In addition to this, for each $x$ we create two integer sequences $I_x,J_x$ of length $2u-1$ each, as follows.
For each $y\in\{0,\ldots,2u-2\}$, we define $(I_x[y],J_x[y])=\Psi(x,y)$ if $\Psi(x,y)\neq \bot$.
If $\Psi(x,y)= \bot$, we choose arbitrary values for $I_x[y]$ and $J_x[y]$ satisfying
$I_x[y]\in \ims(2u-1,u),J_x[y]\in \dms(2u-1,u)$.
From Lemma~\ref{lem-2d-mono}, such sequences $I_x,J_x$ must exist.
By using the data structure $Z(2u-1,u)$ we can store each sequence in at most $\sqrt{2}c^{1/3}u$ bits
and access in constant time.  The total space for these $2m$ sequences is at most
$2um+2\sqrt{2}c^{1/3}um=(2\sqrt{2}c^{1/3}+2)\dis\frac{nm}{t}$ bits.
\end{proof}

\subsection{Proof of Theorem~\ref{th:rmq}}
\begin{proof}
We recursively find range maximum values as in Algorithm~\ref{algo-dfs-rmq}.
Consider the number of times that the function ${\it range}$ of Algorithm~\ref{algo-rec-rmq} is called.
The number of times that an item is added to the set $\ans$ in the function
is at most $\ou$.
On the other hand, if $\ans=\{x_1,\ldots,x_{\ou} \}$, the number of times that
any item is not added to the set $\ans$ is at most $\ou+1$, for ranges 
$[l,x_1-1],[x_1+1,x_2-1],\ldots,[x_{\ou-1}+1,x_{\ou}-1],[x_{\ou}+1,r]$.
Therefore the total number of times that ${\it range}$ is called is at most $2\ou+1$.
From the assumption of the oracle, Algorithm~\ref{algo-dfs-rmq} runs in $\Or(\ou)$ time.
Because the algorithm uses no data structures other than RMQ, the space complexity is $\Or(n)$ bits.
\end{proof}

\section{Omitted Pseudo codes and Figures}

\begin{algorithm}[h]                      
\caption{Obtaining the entry $A[x][y]$ of an two-dimensional array $A\in \tims(n,m)$.}
\label{algo-2d-rec}
	\begin{algorithmic}[1]
	\Require Indices $x,y$
	\Ensure $A[x][y]$
	\State	 $x_b\leftarrow x/t,\ y_b\leftarrow y/t$\Comment{The block number}
	\State	 $x_r\leftarrow x\%t,\ y_r\leftarrow y\%t$\Comment{The position in the block}
	\State ${\it ans}\leftarrow E[x_b][y_b]$
	\If {Block $B_{x_b,y_b}$ is not flat}\label{if-2d-flat}
		\State ${\it ans} \leftarrow {\it ans}+F_{x_b,y_b}[x_r][y_r]$
	\EndIf \\
	\Return ${\it ans}$
	\end{algorithmic}
\end{algorithm}

\begin{algorithm}[h]                         
	\caption{Decide if $B_{i,j}$ is flat or not}
\label{algo-flat} 	
\begin{algorithmic}[1]    
	\Require Block number $i,j$
	\Ensure If block $B_{i,j}$ is flat or not
	\State $x\leftarrow E[i][j],y\leftarrow i-j+u-1$\Comment{$\Phi(i,j)=(x,y)$}
	\If {$K[x][y]=\ $False}\Comment{$\Psi(x,y)=\bot$}
	\State {\bf return} block $B_{i,j}$ is flat
	\EndIf
	\If{$I_x[y]=i$ and $J_x[y]=j$}\Comment{$\Psi(x,y)=(i,j)$}
	\State {\bf return} block $B_{i,j}$ is not flat
	\Else
	\State {\bf return} block $B_{i,j}$ is flat
	\EndIf
	\end{algorithmic}
\end{algorithm}

\begin{algorithm}[h]                      
	\caption{A function to compare $A[i][j]$ with $k$}
	\label{algo-bou}                          
	\begin{algorithmic}[1]    
	\Require an index $(i,j)$ of $A$ and an index $k$ of a boundary
	\Ensure if $A[i][j]\geq k$ or not
	\If {$n-\ra_0(\sel_1(i,B_k),B_k)\leq j$}
	\State  {\bf return} YES
\Else
\State  {\bf return} NO
\EndIf
	\end{algorithmic}
\end{algorithm}

\begin{algorithm}[h]                      
	\caption{Algorithm for Theorem~\ref{thm-loglog2}}         
	\label{algo-loglog}                          
	\begin{algorithmic}[1]    
	\Require a query range $[l,r]$
	\Ensure a mode in range $[l,r]$ of string $S$
	\State $f\leftarrow C[l][r]$ \Comment{using Theorem~\ref{thm-2d-main}}
	\If{$f=1$}
	\State $i \leftarrow l$
	\Else
	\State $i\leftarrow \min\{x\mid C[l][x]=f\}$\Comment{using Theorem~\ref{thm-2d-pred}}
	\EndIf \\
	\Return $S[i]$	
	\end{algorithmic}
\end{algorithm}

\begin{algorithm}[h]                      
	\caption{Algorithm for the range mode enumeration problem using the data structure for finding the leftmost index of range modes}
	\label{algo-enum-left}                          
	\begin{algorithmic}[1]    
	\Require a query range $[l,r]$
	\Ensure the set of all range modes ${\it ans}$
	\State ${\it ans} \leftarrow \{\}$
	\State $(f,i)\leftarrow D([l,r])$\Comment{a pair of the frequency $f$ of modes in range $[l,r]$ and the leftmost index $i$}
	\State $x\leftarrow l$
	\While $\ $\label{algo-whi}
	\State $({\it freq},i)\leftarrow D([x,r])$
	\If {$f>{\it freq}$}\Comment{if the frequency {\it freq} of modes in $[x, r]$ is less than the frequency $f$ of modes in $[l, r]$}
		\State {\bf break}
	\EndIf
	\State ${\it ans}\leftarrow {\it ans}\cup \{S[i]\}$
	\If {$i=r$} \Comment{if the query range becomes empty}
		\State {\bf break}
	\EndIf
	\State $x\leftarrow i+1$\Comment{update the query range}
	\EndWhile\\
	\Return ${\it ans}$
	\end{algorithmic}
\end{algorithm}

\begin{algorithm}[h]                      
	\caption{Find the leftmost index of range modes (assuming $l,r$ belong to different blocks)}         
	\label{algo-left-ind}                          
	\begin{algorithmic}[1]    
	\Require a query range $[l,r]$  ($b_l := \lfloor l/n^{\ep} \rfloor \neq b_r := \lfloor r/n^{\ep} \rfloor$)
	\Ensure $(\mbox{leftmost index} {\it li}, \mbox{frequency} f)$
	\State $f \leftarrow C[b_l][b_r]$\Comment{obtain the frequency of modes of span}
	\State ${\it li} \leftarrow D'[b_l][b_r]$
	\For{$i=l,\ldots,(b_l+1)s-1$} \Comment{check symbols in the prefix}
		\State ${\it cnt}\leftarrow 0$
        \While {$(\mbox{the number of terms of} A[S[i]])>B[i]+{\it f}-1$ and $A[S[i]][B[i]+{\it f}-1]\leq r$} 
			\State ${\it li}\leftarrow \min({\it li},i)$
			\State $f\leftarrow f+1,{\it cnt} \leftarrow {\it cnt}+1$
		\EndWhile
		\If{${\it cnt}>0$}
			\State ${\it f} \leftarrow {\it f}-1$
		\EndIf
	\EndFor
	\For{$i=b_rs,\ldots,r$} \Comment{check symbols in the suffix}
	\State ${\it cnt}\leftarrow 0$
		\While {$0\leq B[i]-{\it freq}+1$ and $A[S[i]][B[i]-{\it freq}+1]\geq l$} 
			\State $f\leftarrow f+1,{\it cnt} \leftarrow {\it cnt}+1$
			\State ${\it li}\leftarrow \min({\it li},A[S[i]][B[i]-{\it freq}+1])$
		\EndWhile	
\If{${\it cnt}>0$}
			\State ${\it f} \leftarrow {\it f}-1$
		\EndIf
\EndFor\\
	\Return $({\it li},f)$
	\end{algorithmic}
\end{algorithm}

\begin{algorithm}[h]                      
	\caption{Find all indices in range $[l,r]$ of sequence $A$ with frequency at least $t$}
	\label{algo-dfs-rmq}                          
	\begin{algorithmic}[1]    
	\Require a query range $[l,r]$, a threshold $t$ 
	\Ensure set of indices $\ans$
	\State $\ans \leftarrow \{\}$
	\State ${\it range}(l,r)$\Comment{the function in Algorithm~\ref{algo-rec-rmq}}\\
	\Return $\ans$
	\end{algorithmic}
\end{algorithm}
\begin{algorithm}[h]                      
	\caption{The recursive function ${\it range}$ called in Algorithm~\ref{algo-dfs-rmq}}
	\label{algo-rec-rmq}                          
	\begin{algorithmic}[1]    
	\Require a range $[x,y]$
	\If {$x>y$}
	\State {\bf return}
	\EndIf 
	\State $id\leftarrow {\it RMQ}(x,y)$
	 \If{$A[id]\geq t$}
	\State $\ans \leftarrow \ans \cup \{id\}$
	\State ${\it range}(x,id-1)$ \Comment{the range to the left of $id$}
	\State ${\it range}(id+1,y)$\Comment{the range to the right of $id$}
	\EndIf
	\end{algorithmic}
\end{algorithm}

\begin{algorithm}[h]                      
	\caption{Algorithm for finding the index set}
	\label{algo-enum}                          
	\begin{algorithmic}[1]    
	\Require a query range $[l,r]$ 
	\Ensure the index set $\ans$
	\State $g\leftarrow f(l,r)$\Comment{using data structure $C$}\label{algo-en21}
\State $b\leftarrow \min\{t\mid f(l,t)\geq g\}$\Comment{using boundary bit-vector $D$}\label{algo-en22}
\State $\ans \leftarrow \{\}$
\State ${\it range}(b,r)$\\
	\Return $\ans$
\\ 
 \State {\bf def function} ${\it range}(x,y)$
 \If{$x>y$}
	\State {\bf return}
\EndIf 
\State ${\it id}\leftarrow RMQ(x,y,H[g])$ \Comment{the index of maximum value in range $[x,y]$ of sequence $H[g]$}
\If{$A[S[id]]-g+1<0$ {\bf and} $B[S[id]][A[S[id]]-g+1]\geq l$}
\State $\ans\leftarrow \ans \cup \{\it id\}$
	\State ${\it range}(x,id-1)$ \Comment{the range to the left of $id$}
	\State ${\it range}(id+1,y)$\Comment{the range to the right of $id$}
\EndIf
 \State {\bf end def}
	\end{algorithmic}
\end{algorithm}

\begin{algorithm}[h]                      
	\caption{Algorithm for finding the mode index set}
	\label{algo-n2bits}                          
	\begin{algorithmic}[1]    
	\Require a query range $[l,r]$ 
	\Ensure the mode index set $\ans$
	\State $g\leftarrow f(l,r)$\label{algo-n21}
\State $b\leftarrow \min\{t\mid f(l,t)\geq g\}$\label{algo-n22}
\State $x\leftarrow r$
\State $\ans \leftarrow \{\}$
	\While{$x\geq b$} \label{algo-n2while}
\State $\ans \leftarrow \ans \cup \{x\}$\Comment{add to the mode index set}
\State $x\leftarrow \sel_1(\ra_1(x-1,B[l]),B[l])$\Comment{update $x$}
\EndWhile\\
	\Return $\ans$
	\end{algorithmic}
\end{algorithm}

	\begin{figure}[h]
		\begin{center}
			\includegraphics[width=90mm]{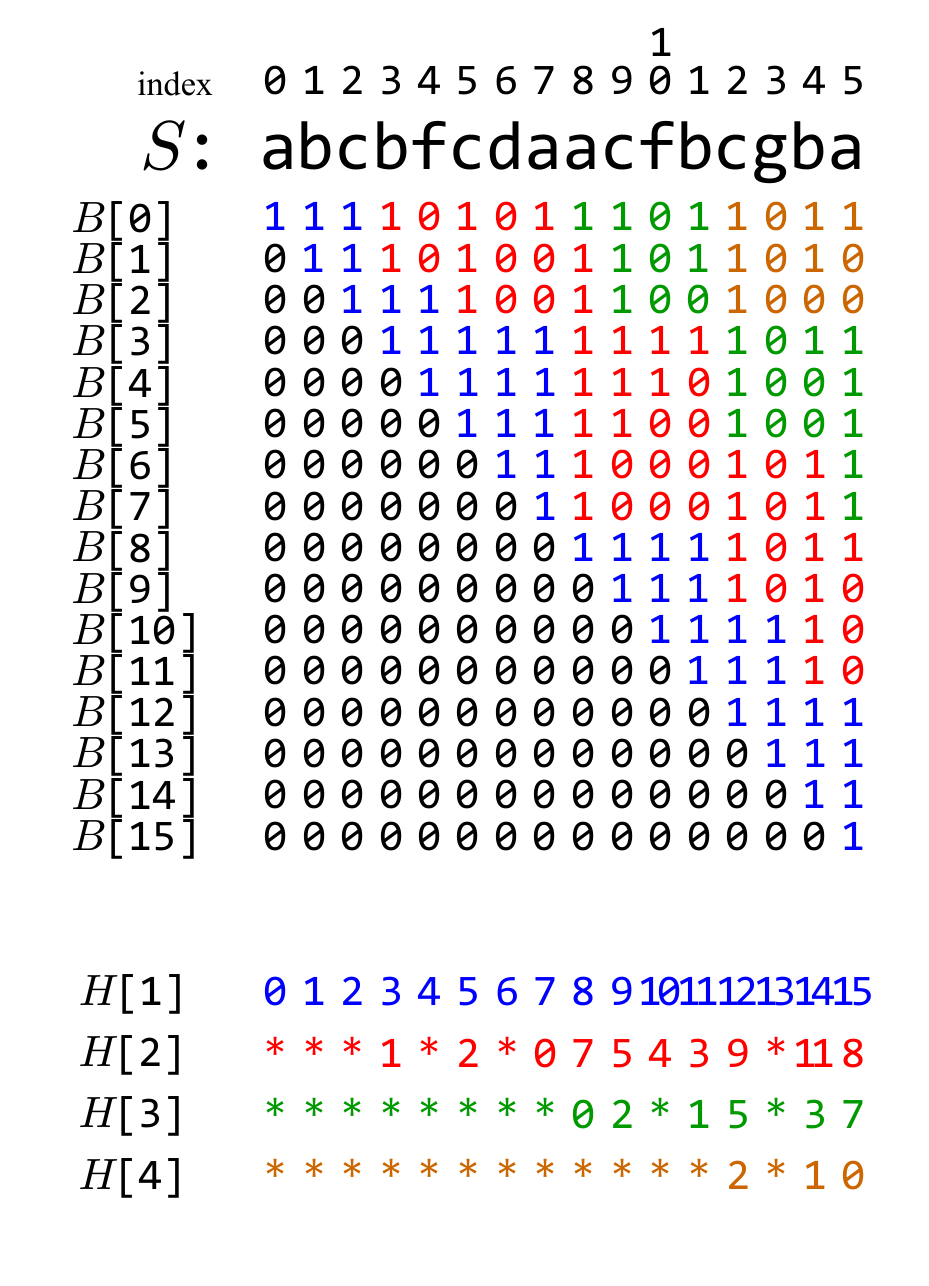}
			\end{center}
		\caption{Bit-vectors $B[0],\ldots,B[15]$ and sequences $H[1],\ldots,H[4]$ for string $S$ of length $n=16$.  The marks * stand for $-1$.
		Colors of numbers for $B$ represent frequencies of modes of the corresponding ranges.
		Blue, red, green, and brown colors represent frequencies 1, 2, 3, and 4, respectively.}
		\label{fig-nmbit}
	\end{figure}

\end{document}